\documentclass[preprint,
]{revtex4-1}
\usepackage{graphicx}
\usepackage{amsfonts}
\usepackage{amssymb}
\usepackage{amsthm}
\usepackage{array}
\usepackage{amsmath}
\usepackage{verbatim} 
\usepackage{hyperref}
\usepackage{color}
\usepackage{bbold}
\usepackage{epstopdf}
\usepackage{mathtools}
\usepackage{enumerate}
\usepackage{subcaption}
\usepackage{color}

\newtheorem{proposition}{Proposition}
\newtheorem*{proposition*}{Proposition}
\newtheorem{definition}{Definition}
\newtheorem{theorem}{Theorem}
\newtheorem*{theorem*}{Theorem}

\newtheorem*{corollary*}{Corollary}

\newcommand{\ket}[1]{\left\vert#1\right\rangle}
\newcommand{\bra}[1]{\left\langle#1\right\vert}

\def\Tr{\mathrm{Tr}}

\begin{document}
\title{Entanglement as the symmetric portion of correlated coherence}
\author{Kok Chuan Tan}
\email{bbtankc@gmail.com}
\author{Hyunseok Jeong}
\email{jeongh@snu.ac.kr}
\affiliation{Center for Macroscopic Quantum Control \& Institute of Applied Physics, Department of Physics and Astronomy, Seoul National University, Seoul, 151-742, Korea}
\date{\today}

\begin{abstract}
We show that the symmetric portion of correlated coherence is always a valid quantifier of entanglement, and that this property is independent of the particular choice of coherence measure. This leads to an infinitely large class of coherence based entanglement monotones, which is always computable for pure states if the coherence measure is also computable. It is already known that every entanglement measure can be constructed as a coherence measure. The results presented here show that the converse is also true. The constructions that are presented can also be extended to also include more general notions of nonclassical correlations, leading to quantifiers that are related to quantum discord.
\end{abstract}

\maketitle

\section{Introduction}

An important pillar in the field of quantum information is the study of the quantumness of correlations, the most well known of which is the notion of entangled quantum states~\cite{Einstein1935}. Entanglement is now the basis of many of the most useful and powerful quantum protocols, such as quantum cryptography~\cite{Ekert1991}, quantum teleportation~\cite{Bennett1991} and superdense coding. In the past several decades, generalized notions of quantum correlations that include but supersede entanglement have also been considered, most prominently in the form quantum discord~\cite{Ollivier2001, Henderson2001}. There is mounting evidence that such notions of quantum correlations can also lead to nonclassical effects in multipartite scenarios~\cite{Datta2008, Chuan2012, Dakic2012, Chuan2013}, even when entanglement is not available.

In a separate development, the past several years has also seen a growing amount of interest in the recently formalized resource theory of coherence~\cite{Aberg2006,Baumgratz2014, Levi2014}. Such theories are primarily interested in identifying the quantumness of some given quantum state, and is not limited to a multipartite setting as in the case of entanglement or discord. Nonetheless, there is considerable interest in the study of correlations from the point of view of coherence~\cite{Streltsov2015, Tan2016, Ma2016, Tan2018}. In this picture, one may view quantum correlations as a single aspect of the more general notion of nonclassicality, which in this article we will assume to imply coherence. Beyond the study of quantum correlations, the resource theory of coherence has been applied to an ever increasing number of physical scenarios, ranging from macroscopicity~\cite{Yadin2016, Kwon2017}, to quantum algorithms~\cite{Hillery2016,Matera2016}, to interferometry~\cite{YTWang2017}, to nonclassical light~\cite{Tan2017, Zhang2015, Xu2016}. \cite{Streltsov2017} provides a recent overview of the developments to date. Especially relevant are the results in~\cite{Streltsov2015}. There, it was shown that coherence can be faithfully converted into entanglement, and that each entanglement measure corresponds to a coherence measure in the sense of \cite{Baumgratz2014}.

In this article, we report a series of constructions which allows notions of nonclassical correlations to be quantified using coherence measures. The arguments are general in the sense that it does not depend on the particular coherence measure used, and does not even depend on the particular flavour of coherence measure that is being employed, so long as they satisfy some minimal set of properties that a reasonable coherence measure should satisfy. This suggests that notions of entanglement and discord are intrinsically tied to any reasonable resource theory of coherence. In essence, our results establish that the converse of the relationship proposed in~\cite{Streltsov2015} is also true, so that for every coherence measure, there corresponds an entanglement measure. In fact, we also go beyond this by demonstrating that not only does this correspondence exist for the coherence resource theory proposed in \cite{Baumgratz2014}, as our framework does not depend on the choice of free operations required by a particular resource theory \cite{Streltsov2017}. In addition, as natural consequences of our  operation-independent approach, we also see that discord-like quantifiers of nonclassical correlations are naturally embedded in any such resource theories of coherence.

This operation-free approach contrasts with other approaches considered in~\cite{Chitambar2016, Streltsov2016}, where coherence and entanglement are bridged by forming a hybrid resource theory based on some combination of free operations from both theories. Such a hybridization approach often require additional constraints, such as requiring that operations are both local on top of being incoherent, which may bring about extra complications. For instance, it is sometimes difficult to physically justify the set of operations being considered in the hybridized theory, and one may also have to deal with the accounting of not one, but two, species of resource states (i.e. simultaneously keep track of available maximally coherent qubits on top of maximally entangled qubits). 

\section{Preliminaries}

We review some elementary concepts concerning coherence measures. Coherence is a basis dependent property of a quantum state. For a given fixed basis $\mathcal{B} = \{ \ket{i} \}$, the set of incoherent states $\cal I$ is the set of quantum states with diagonal density matrices with respect to this basis, and is considered to the the set of classical states. Correspondingly, states that have nonzero off diagonal elements form the set of coherent states that are nonclassical.

The notion of nonclassicality from the point of view of coherence is an unambiguous aspect of all coherence theories, but different flavors of coherence resource theories sometimes consider different sets of non-coherence producing operations in order to justify different coherence measures (See~\cite{Streltsov2017} for a summary). For our purposes, we will not require any specific properties of such non-coherence producing operations. The following is a  set of axioms that such resource theories of coherence generally obeys: Let $\mathcal{C}$ be a measure of coherence belonging to some coherence resource theory, then $\mathcal{C}(\rho)$ must satisfy
(C1) $\mathcal{C}(\rho) \geq 0$ for any quantum state $\rho$ and equality holds if and only if $\rho \in \cal I$.
(C2) The measure is non-increasing under a non-coherence producing map $\Phi$ , i.e., $C(\rho) \geq C(\Phi(\rho))$.
(C3) Convexity, i.e. $\lambda C(\rho) + (1-\lambda) C(\sigma) \geq C(\lambda \rho + (1-\lambda) \sigma)$, for any density matrix $\rho$ and $\sigma$ with $0\leq \lambda \leq 1$.

The following quantity was considered by Tan {\it et al.}~\cite{Tan2016} while studying the relationship between coherence and quantum correlations: $$\mathcal{C}(A:B \mid \rho_{AB}) \coloneqq \mathcal{C}(\rho_{AB}) - \mathcal{C}(\rho_{A}) - \mathcal{C}(\rho_{B}).$$

This quantity was referred to as correlated coherence, and the coherence measure $\mathcal{C}$ in \cite{Tan2016} was chosen to be $l_1$-norm of coherence. There, it was noted that since it is always possible to choose a local basis for the subsystems $A$ and $B$ where $\mathcal{C}(\rho_{A})$ and $\mathcal{C}(\rho_{A})$ vanishes, the coherence in the system is apparently no longer stored locally, and must exist amongst the correlations between subsystems $A$ and $B$. 

Based on this observation it was demonstrated there that if one were to minimize the correlated coherence with respect to all such possible local bases, i.e. all possible local bases $\mathcal{B}_A$ and $\mathcal{B}_B$ satisfying $\mathcal{C}(\rho_{A})=\mathcal{C}(\rho_{A})=0$, then the minimization over all such bases may be related to quantum correlations such as discord and entanglement. Formally, they considered the quantity:

\begin{definition}[Correlated coherence] $$\mathcal{C}_{\mathrm{min}}(A:B \mid \rho_{AB}) \coloneqq \min_{\mathcal{B}_{A:B}}\mathcal{C}(A:B \mid \rho_{AB}),$$ where the minimization is performed over the set of local bases $\mathcal{B}_{A:B} \coloneqq \{ (\mathcal{B}_A, \mathcal{B}_B) \mid \mathcal{C}(\rho_{A})=\mathcal{C}(\rho_{B})=0 \}.$ \end{definition}

 The quantity is invariant under local unitary operations, since it is clear that for any state $\rho_{AB}$ and local basis $\mathcal{B}_{A:B} = \{\ket{i}_A\ket{j}_B \}$, the correlated coherence for the state $U_A \rho_{AB} U^\dag_A$ and basis $\mathcal{B}_{A:B} = \{U_A\ket{i}_A\ket{j}_B \}$ is identical. Subsequently, entropic versions of correlated coherence was also studied in \cite{Wang2017} and more recently in \cite{Kraft2018, Ma2018}, where operational scenarios were considered.

In the next section, we prove that using $\mathcal{C}_{\mathrm{min}}(A:B \mid \rho_{AB})$ as our basic building block, every coherence measure can be used to construct a valid entanglement quantifier, which establishes that entanglement may be interpreted as the symmetric portion of correlated coherence. 

\section{Quantifying entanglement with correlated coherence}

We begin with some necessary definitions:

\begin{definition}[Symmetric extensions]
A symmetric extension of a bipartite state $\rho_{A_1B_1}$ is an extension $\rho_{A_1\ldots A_n B_1\ldots B_n}$ satisfying $\Tr_{A_2 \ldots A_n B_2 \ldots B_n}(\rho_{A_1\ldots A_n B_1\ldots B_n}) = \rho_{A_1B_1}$ that is, up to local unitaries, invariant under the swap operation $\Phi_{\mathrm{SWAP}}^{A_i \leftrightarrow B_i}$ between any subsystems $A_i$ of Alice and $B_i$ of Bob, i.e. there exists some unitary $U_{A_1 \ldots A_n}$ such that $$\Phi_{\mathrm{SWAP}}^{A_i \leftrightarrow B_i}(U_{A_1 \ldots A_n}\rho_{A_a \ldots A_n B_1\ldots B_n}U^\dag_{A_1 \ldots A_n}) = U_{A_1 \ldots A_n}\rho_{A_a \ldots A_n B_1\ldots B_n}U^\dag_{A_1 \ldots A_n}$$

\end{definition}

A symmetric extension is therefore, up to a local unitary on Alice's side (or Bob's side), an extension of the quantum state that, up to local unitaries, exists within the symmetric subspace. Subsequently, for notational simplicity, we will use unprimed letters $A,B$ for the system of interest, and and primed letters $A', B'$ for the ancillas in the extension. Let us now consider the correlated coherence for such extensions.

\begin{definition}[Symmetric correlated coherence]
The symmetric correlated coherence, for any given coherence measure $\mathcal{C}$, is defined to be the following quantity: $$E_{\mathcal{C}}(\rho_{AB}) = \min_{A'B'}\mathcal{C}_{\mathrm{min}}(AA':BB' \mid \rho_{AA'BB'})$$ where the minimization is performed over all possible symmetric extensions of $\rho_{AB}$.  Note that the ancillas $A'$ and $B'$ may, in general, be composite systems.
\end{definition}

The above definition quantifies the minimum correlated coherence that exists within a symmetric subspace of an extended Hilbert space, up to some local unitary on Alice's side or Bob's side. For this reason, we interpret this quantity as the portion of the correlated coherence that is symmetric.

For the rest of the note, we will prove several elementary properties of the above correlation measure, which will finally establish it as a valid entanglement monotone. 

First, we will demonstrate that $E_{\mathcal{C}}(\rho_{AB})$ is a convex function of states:

\begin{proposition}
$E_{\mathcal{C}}(\rho_{AB})$ is a convex function of state, i.e. $$\sum_i p_i E_{\mathcal{C}}(\rho^i_{AB}) \geq E_{\mathcal{C}}(\sum_i p_i\rho^i_{AB})$$ where $p_i$ defines some probability distribution s.t. $\sum_i p_i = 1$ and $\rho^i_{AB}$ is any normalized quantum state.
\end{proposition}

\begin{proof}
Let $\rho^{i*}_{AA'BB'}$ be the optimal extension such that $E_{\mathcal{C}}(\rho^i_{AB})= \mathcal{C}_{\mathrm{min}}(AA':BB' \mid \rho^{i*}_{AA'BB'})$. We have the following chain of inequalities:

\begin{align}
\sum_i p_i E_{\mathcal{C}}(\rho^i_{AB})& =  \sum_i p_i \mathcal{C}_{\mathrm{min}}(AA':BB' \mid \rho^{i*}_{AA'BB'}) \\
&= \sum_i p_i \mathcal{C}_{\mathrm{min}}(AA'A'':BB'B'' \mid \rho^{i*}_{AA'BB'}\otimes \ket{i,i}_{A''B''}\bra{i,i}) \\
&\geq  \mathcal{C}_{\mathrm{min}}(AA'A'':BB'B'' \mid \sum_i p_i \rho^{i*}_{AA'BB'}\otimes \ket{i,i}_{A''B''}\bra{i,i}) \\
&\geq  E_\mathcal{C}( \sum_i p_i \rho^i_{AB} )
\end{align}  

The inequality in Line 3 occurs because there is at least one local basis that is upper bounded by Line 2.  To see this, suppose for every $i$ and $\rho^{i*}_{AA'BB'}$, the optimal basis for evaluating $\mathcal{C}_{\mathrm{min}}(AA':BB' \mid \rho^{i*}_{AA'BB'})$ is $\{ \ket{\alpha_{i,j}}_{AA'} \ket{\beta_{i,k}}_{BB'} \}$. Then it is clear that the optimal local basis for $\rho^{i*}_{AA'BB'}\otimes \ket{i,i}_{A''B''}\bra{i,i}$ must be $\{ \ket{\alpha_{i,j}}_{AA'} \ket{i}_{A''}\ket{\beta_{i,k}}_{BB'} \ket{i}_{B''} \}$ since there was just essentially a relabelling of the basis. Since the coherence measure $\mathcal{C}$ is convex, the classical mixture of quantum states cannot increase the amount of coherence with respect to the basis $\{ \ket{\alpha_{i,j}}_{AA'} \ket{i}_{A''}\ket{\beta_{i,k}}_{BB'} \ket{i}_{B''} \}$. Finally, one can verify that the local coherences with respect to this basis is always zero, so this is just one particular local basis that satisfies the necessary contraints. In sum, this implies \begin{align*}\sum_i p_i &\mathcal{C}_{\mathrm{min}}(AA'A'':BB'B'' \mid \rho^{i*}_{AA'BB'}\otimes \ket{i,i}_{A''B''}\bra{i,i}) \\
&\geq  \mathcal{C}_{\mathrm{min}}(AA'A'':BB'B'' \mid \sum_i p_i \rho^{i*}_{AA'BB'}\otimes \ket{i,i}_{A''B''}\bra{i,i}),\end{align*} which was the required inequality.

The inequality in Line 4 comes from the observation that $\sum_i p_i \rho^*_{AA'BB'}\otimes \ket{i,i}_{A''B''}\bra{i,i}$ is a particular symmetric extension of $\sum_i p_i \rho^*_{AB}$. The final inequality is simply the condition of convexity which we needed to prove.
\end{proof}

In the next proposition, we demonstrate the connection between $E_{\mathcal{C}}(\rho_{AB})$ and nonseparability, which defines entanglement.

\begin{proposition} [Faithfulness]
 $E_{\mathcal{C}}(\rho_{AB}) = 0$ iff $\rho_{AB}$ is separable, and strictly positive otherwise.
\end{proposition}

\begin{proof}
First of all, we note that all coherence measures are nonnegative over valid quantum states, and as such, since $E_{\mathcal{C}}(\rho_{AB})$ is defined as a form of coherence over some quantum state, $E_{\mathcal{C}}(\rho_{AB})$ must be nonnegative.

Suppose some bipartite state $\rho_{AB}$ is separable. By definition, this necessarily implies that there exists some decomposition for which $\rho_{AB} = \sum_i p_i \ket{a_i}_A\bra{a_i} \otimes \ket{b_i}_B\bra{b_i}$. This always permits an extension of the form $\rho_{AA'BB'} = \sum_i p_i \ket{a_i}_A\bra{a_i} \otimes \ket{i}_{A'}\bra{i} \otimes  \ket{b_i}_B\bra{b_i} \otimes \ket{i}_{B'}\bra{i}$ for some orthonormal set $\{ \ket{i} \}$. 
It can then be directly verified that $\mathcal{C}_{\mathrm{min}}(AA':BB' \mid  \rho_{AA'BB'}) = 0$ so we must have $E_{\mathcal{C}}(\rho_{AB})=0$ for every separable state.

We now prove the converse. Suppose $E_{\mathcal{C}}(\rho_{AB})=0$. Then there must exist some extension for which  $\mathcal{C}_{\mathrm{min}}(AA':BB' \mid  \rho_{AA'BB'}) = 0$. This implies that there must exist a local basis on $AA'$ and on $BB'$ such that the coherence must be zero, so $\rho_{AA'BB'}$ necessarily must be diagonal in this basis, i.e. $\rho_{AA'BB'} = \sum_{i} q_i \ket{\alpha_i}_{AA'}\bra{\alpha_i} \otimes \ket{\beta_i}_{BB'}\bra{\beta_i}$. Directly tracing out the subsystems $A'$ and $B'$ will lead to a decomposition of the form $\rho_{AB} = \sum_i p_i \ket{a_i}_A\bra{a_i} \otimes \ket{b_i}_B\bra{b_i}$, so $\rho_{AB}$ must be a separable state.

We then observe that since $E_{\mathcal{C}}(\rho_{AB})$ must be nonnegative, and it is zero iff $\rho_{AB}$ is separable, then it must be strictly positive for every entangled state. This completes the proof.
\end{proof}

Finally, we show that $E_{\mathcal{C}}(\rho_{AB})=0$ always decreases under LOCC type operations.

\begin{proposition} [Monotonicity] \label{thm::monotonicity} For any LOCC protocol represented by a quantum map $\Phi_{\mathrm{LOCC}}$, we have $$E_{\mathcal{C}}( \rho_{AB}) \geq E_{\mathcal{C}}[\Phi_{\mathrm{LOCC}}(\rho_{AB})].$$
\end{proposition}

\begin{proof}
Any LOCC operation can always be decomposed into some local quantum operation, a communication of classical information stored in a classical register, and finally, another local operation that is dependent on the classical information received.

Let us suppose that Bob, representing the subsystem $B$ is the one who will communicate classical information to Alice, representing subsystem $A$. His local operation can always be represented by adding ancillary subsystems $B'B''$ in some initial pure state $\ket{0}_{B'}\bra{0} \otimes \ket{0}_{B''}\bra{0}$, followed by a unitary operation on all of the subsystems on his side. Without any loss in generality, we will assume $B''$ will contain all the classical information (i.e. it is a classical register) after the unitary is performed, and $B'$ is traced out. Bob will then communicate this classical information to Alice, who will then perform some quantum operation depending on the information she received.

Based on the above, we have the following chain in inequalities:

\begin{align}
E_{\mathcal{C}}(\rho_{AB})&=  E_{\mathcal{C}}( \rho_{AB} \otimes \ket{0}_{B'}\bra{0} \otimes \ket{0}_{B''}\bra{0}) \\
&=  E_{\mathcal{C}}( U_{BB'B''}\rho_{AB} \otimes \ket{0}_{B'}\bra{0} \otimes \ket{0}_{B''}\bra{0}U_{BB'B''}^\dag) \\
&\geq E_\mathcal{C}(\sum_i K_B^i \rho^{*}_{AB}K_B^{i\dag} \otimes \ket{i}_{B''}\bra{i})
\end{align} where the last line makes use of the observation that a symmetric extension of the argument in Line 6 is also a symmetric extension of the argument of Line 7.

From the above, we see that a local POVM performed on Bob's side necessarily decreases the $E_\mathcal{C}$. The next part of the protocol requires Bob to communicate the classical information in the register $B''$ over to Alice. We need to demonstrate that this can be done for free, without increasing the correlated coherence. 

To see this, let $\sigma_{AA'A''BB'B''}^*$ be the optimal symmetric extension of $\sum_i K_B^i \rho^{*}_{AB}K_B^{i\dag} \otimes \ket{i}_{B''}\bra{i}$. We then have $E_\mathcal{C}(\sum_i K_B^i \rho^{*}_{AB}K_B^{i\dag} \otimes \ket{i}_{B''}\bra{i}) = \mathcal{C}_{\mathrm{min}}(AA'A'':BB'B'' \mid \sigma_{AA'A''BB'B''}^*)$. Recall that the register $B''$ stores the classical information of Bob's POVM outcomes. By definition, $\sigma_{AA'A''BB'B''}^*$ must be a symmetric extension, so there exists a local unitary that Alice can perform such that $U_{AA'A''}\sigma_{AA'A''BB'B''}^* U_{AA'A''}^\dag = \Phi_{\mathrm{SWAP}}^{AA'A'' \leftrightarrow BB'B''}(U_{AA'A''}\sigma_{AA'A''BB'B''}^* U_{AA'A''}^\dag)$. Since local unitaries do no affect the measure $E_\mathcal{C}$, so we will assume that $\sigma_{AA'A''BB'B''}^*$ is itself already symmetric.

Suppose we add add registers, denoted $M_{A}$ and $M_{B}$, initialized in the state $\ket{0}_{M_A}$ and $\ket{0}_{M_B}$, and locally copy the classical information on  registers $A''$ and $B''$ via CNOT operations $U_{\mathrm{CNOT}}^{M_AA''}$ and $U_{\mathrm{CNOT}}^{M_BB''}$. This results in the state $$\mathcal{U}_{\mathrm{CNOT}}^{M_AA''} \circ \mathcal{U}_{\mathrm{CNOT}}^{M_BB''}(\ket{0}_{M_A}\bra{0} \otimes  \sigma_{AA'A''BB'B''}^* \otimes \ket{0}_{M_B}\bra{0}),$$ where $\mathcal{U}_{\mathrm{CNOT}}^{AB}(\rho_{AB}) \coloneqq U_{\mathrm{CNOT}}^{AB} \rho_{AB} U_{\mathrm{CNOT}}^{AB \dag}$. Note that as identical unitary operations are performed on Alice's and Bob's side, the above state is symmetric since $\sigma_{AA'A''BB'B''}^*$ is symmetric. Due to symmetry, we must have the following chain of equalities:

\begin{align} &\mathcal{U}_{\mathrm{CNOT}}^{M_AA''} \circ \mathcal{U}_{\mathrm{CNOT}}^{M_BB''}(\ket{0}_{M_A}\bra{0} \otimes  \sigma_{AA'A''BB'B''}^* \otimes \ket{0}_{M_B}\bra{0}) \\& \;= \Phi_{\mathrm{SWAP}}^{A''\leftrightarrow B''} \circ \mathcal{U}_{\mathrm{CNOT}}^{M_AA''} \circ \mathcal{U}_{\mathrm{CNOT}}^{M_BB''}[\ket{0}_{M_A}\bra{0} \otimes  \Phi_{\mathrm{SWAP}}^{A''\leftrightarrow B''} (\sigma_{AA'A''BB'B''}^*) \otimes \ket{0}_{M_B}\bra{0}] \\ & \;= \mathcal{U}_{\mathrm{CNOT}}^{M_AB''} \circ \mathcal{U}_{\mathrm{CNOT}}^{M_BA''}(\ket{0}_{M_A}\bra{0} \otimes  \sigma_{AA'A''BB'B''}^* \otimes \ket{0}_{M_B}\bra{0})\end{align} where Equation~10 uses the fact that $\Phi_{\mathrm{SWAP}}^{A\leftrightarrow B}(\rho_{AB}) = U_{\mathrm{SWAP}}^{A\leftrightarrow B} \rho_{AB} U_{\mathrm{SWAP}}^{A\leftrightarrow B\;  \dag}$, $ U_{\mathrm{SWAP}}^{A\leftrightarrow B} =  U_{\mathrm{SWAP}}^{A\leftrightarrow B\;  \dag}$ and $U_{\mathrm{SWAP}}^{B\leftrightarrow C} U_{\mathrm{CNOT}}^{AB} U_{\mathrm{SWAP}}^{B\leftrightarrow C \; \dag} = U_{\mathrm{CNOT}}^{AC}$. One can verify that Equation~10 is a symmetric extension of $\sum_i \ket{i}_{M_A}\bra{i}\otimes K_B^i \rho^{*}_{AB}K_B^{i\dag} \otimes \ket{i}_{B''}\bra{i}$, which is just the state if Bob communicates the classical information on the register $B''$ to Alice. As such, we determine that the copying of classical information to Alice cannot increase the measure, so we have $$E_\mathcal{C}(\sum_i K_B^i \rho^{*}_{AB}K_B^{i\dag} \otimes \ket{i}_{B''}\bra{i}) \geq E_\mathcal{C}(\sum_i \ket{i}_{M_A}\bra{i}\otimes K_B^i \rho^{*}_{AB}K_B^{i\dag} \otimes \ket{i}_{B''}\bra{i}).$$ This is already sufficient for us to prove that $E_\mathcal{C}$ cannot increase under classical communication.

Continuing from where we left off:

\begin{align}
E_{\mathcal{C}}(\rho_{AB})& \geq E_\mathcal{C}(\sum_i K_B^i \rho^{*}_{AB}K_B^{i\dag} \otimes \ket{i}_{B''}\bra{i}) \\ 
&\geq E_\mathcal{C}(\sum_i  \ket{i}_{A''}\bra{i} \otimes K_B^i \rho^{*}_{AB}K_B^{i\dag} \otimes \ket{i}_{B''}\bra{i}) \\ 
&= E_\mathcal{C}(\ket{0}_{A'}\bra{0} \otimes \sum_i  \ket{i}_{A''}\bra{i} \otimes K_B^i \rho^{*}_{AB}K_B^{i\dag} \otimes \ket{i}_{B''}\bra{i}) \\ 
&= E_\mathcal{C}(U_{AA'A''}\ket{0}_{A'}\bra{0} \otimes \sum_i  \ket{i}_{A''}\bra{i} \otimes K_B^i \rho^{*}_{AB}K_B^{i\dag} \otimes \ket{i}_{B''}\bra{i} U_{AA'A''}^\dag) \\
 &\geq E_\mathcal{C}( \sum_{i,j}  \ket{i}_{A''}\bra{i} \otimes K_A^{i,j}K_B^i \rho^{*}_{AB}K_B^{i\dag} K_A^{i,j\dag} \otimes \ket{i}_{B''}\bra{i} )
\end{align} where in Line 11 and 12, we used the fact that local operations and lassical communication cannot increase $E_\mathcal{C}$, and in Line 15, the inequality is because every symmetric extension of the argument in Line 14 is also a symmetric extension of the argument in line 15. The final line says that when Alice performs an operation conditioned on the classical communication by Bob, the measure also does not increase.

From the above arguments, we see that any local POVM performed by Bob, followed by a communication of classical measurement outcomes to Alice, and ended by another local quantum operation by Alice conditioned on the classical communication necessarily cannot increase $E_\mathcal{C}$. Since any LOCC protocol is a series of such procedures between Alice and Bob, possibly with their roles reversed, this implies that $E_\mathcal{C}$ is always contractive under LOCC operations.
 
\end{proof}

In sum, the Propositions directly imply the following theorem, which is the key result of this article:
\begin{theorem}
$E_\mathcal{C}$ is a valid entanglement monotone for every choice of coherence measure $\mathcal{C}$.
\end{theorem}

We observe that if we were to choose the coherence measure to be the relative entropy of coherence, which is defined as $\mathcal{C}(\rho_{AB}) = \mathcal{S}[\Delta(\rho_{AB})]- \mathcal{S}(\rho_{AB})$ where $\Delta(\rho_{AB})$ is the completely dephased state~\cite{Baumgratz2014}, then for pure states, the measure exactly coincides with the well known entropy of entanglement. This is because pure quantum states only have trivial extensions and it always permits Schmidt decomposition of the state $\ket{\psi}_{AB} = \sum_i\sqrt{\lambda_i}\ket{i,i}_{AB}$ where we observe that the local bases $\{\ket{i}_A \}$ and $\{ \ket{i}_B \}$ satisfies the condition that $\mathcal{C}(\rho_A)=\mathcal{C}(\rho_B)=0$. We can then verify w.r.t. this basis, $\mathcal{S}[\Delta(\rho_{AB})]= \mathcal{S}(\sum_i{\lambda_i \ket{i,i}_{AB}\bra{i,i}}) = \mathcal{S}[\Tr_B(\ket{\psi}_{AB}\bra{\psi})]$ which is just the expression for the entropy of entanglement. It still remains to be proven that the local bases $\{\ket{i}_A \}$ and $\{ \ket{i}_B \}$ achieves the minimization required in the definition of correlated coherence. In fact, this is always true and is a generic property of all coherence measures, which we show in following theorem.

\begin{theorem}[$E_{\mathcal{C}}$ for pure states] \label{thm::purestates}
For any continuous coherence measure $\mathcal{C}$ and pure state $\ket{\psi}_{AB}$ with Schmidt decomposition $\ket{\psi}_{AB} = \sum_i\sqrt{\lambda_i}\ket{i,i}_{AB}$, $E_\mathcal{C}(\ket{\psi}_{AB}) = \mathcal{C}(\ket{\psi}_{AB})$ where the coherence is measured w.r.t. the local bases $\{\ket{i}_A \}$ and $\{ \ket{i}_B \}$ specified by the Schmidt decomposition.
\end{theorem}

\begin{proof}
Consider the Schmidt decomposition $\ket{\psi}_{AB} = \sum_i\sqrt{\lambda_i}\ket{i,i}_{AB}$ for any pure state. We don't need to consider extensions since every pure state only has trivial extensions. Suppose the coefficients are nondegenerate, in the sense that $\lambda_i \neq \lambda_j$ if $i\neq j$. If we were to perform a partial trace, we see that $\rho_A = \Tr_B(\ket{\psi}_{AB}\bra{\psi})= \sum_i \lambda_i \ket{i}_A\bra{i}$. As the the coefficients are nondegenerate, this implies that $\{\ket{i}_A\}$ (up to an overall phase factor) is the unique local basis satisfying $\mathcal{C}(\rho_A)=0$.  Identical arguments also apply for subsystem $B$. As such, the local bases $\{\ket{i}_A \}$ and $\{ \ket{i}_B \}$ necessarily achieves the minimum for the correlated coherence, i.e. $\mathcal{C}_{\mathrm{min}}(A:B \mid \ket{\psi}_{AB})$ and $E_\mathcal{C}(\ket{\psi}_{AB})$ are just the coherence $\mathcal{C}(\ket{\psi}_{AB})$ w.r.t. the local bases $\{\ket{i}_A \}$ and $\{ \ket{i}_B \}$.

The above demonstrates that the local bases defined by the Schmidt decomposition achieves the necessary minimization when the coefficients are nondegenerate. We now extend the arguments to the more general case. Consider now a general Schmidt decomposition $\ket{\psi}_{AB} = \sum_i\sqrt{\lambda_i}\ket{i,i}_{AB}$. In this case, even if the coefficients are degenerate, the local bases $\{\ket{i}_A \}$ and $\{ \ket{i}_B \}$ nonetheless satisfies the constraints  $\mathcal{C}(\rho_A)=\mathcal{C}(\rho_B)=0$ so $\mathcal{C}(\ket{\psi}_{AB})$ w.r.t. this basis is at least an upper bound to $E_{\mathcal{C}}(\ket{\psi}_{AB})$. 

Consider again the partial trace $\rho_A = \sum_i \lambda_i \ket{i}_A\bra{i}$. Without any loss in generality, we will assume that $\lambda_i$ is in decreasing order, so that if $j\geq i$ then $\lambda_j \leq \lambda_i$. Suppose one of the coefficient is $m$-degenerate, so that for some $k$, $\lambda_k > \lambda_{k+1}= \ldots= \lambda_{k+m} > \lambda_{k+m+1}$. Note the strict inequality on both ends. We now consider a slightly perturbed state $\ket{\psi(\epsilon)}_{AB} = \sum_i\sqrt{\lambda_i(\epsilon)}\ket{i,i}_{AB}$ where $\lambda_i(\epsilon) = \lambda_i-\epsilon \lfloor i- k - \frac{m}{2} \rfloor$ whenever $k<i<k+m+1$ and $\lambda_i(\epsilon) = \lambda_i$ otherwise. The corresponding partial trace is denoted $\rho_A(\epsilon) = \sum_i \lambda_i(\epsilon) \ket{i}_A\bra{i}$. For sufficiently small $\epsilon >0$, we can verify that the majorization condition $\rho_A \prec \rho_A(\epsilon)$ is satisfied, which due to Nielsen's theorem \cite{Nielsen2010}, implies that there exists some LOCC operation $\Phi_{\mathrm{LOCC}}$ that performs the transformation  $\ket{\psi}_{AB} \rightarrow\ket{\psi(\epsilon)}_{AB}$. From Theorem~\ref{thm::monotonicity}, we know that this implies $E_\mathcal{C}(\ket{\psi}_{AB}) \geq E_\mathcal{C}\left(\ket{\psi(\epsilon)}_{AB} \right)$ since the quantity cannot increase under LOCC operations. At the same time, for sufficiently small $\epsilon>0$, the coefficients $\lambda_i(\epsilon)$ are non-degenerate so $E_\mathcal{C}(\ket{\psi(\epsilon)}_{AB}) = \mathcal{C}(\ket{\psi(\epsilon)}_{AB})$ w.r.t. the local bases $\{\ket{i}_A \}$ and $\{ \ket{i}_B \}$. This implies that $\mathcal{C}(\ket{\psi(\epsilon)}_{AB}) \leq E_\mathcal{C}(\ket{\psi}_{AB}) \leq \mathcal{C}(\ket{\psi}_{AB})$. In the limit $\epsilon \rightarrow 0$, $\mathcal{C}(\ket{\psi(\epsilon)}_{AB})\rightarrow \mathcal{C}(\ket{\psi}_{AB})$ , so by the squeeze theorem we must have that $E_\mathcal{C}(\ket{\psi}_{AB}) = \mathcal{C}(\ket{\psi}_{AB})$, where the implied basis is given by  $\{\ket{i}_A \}$ and $\{ \ket{i}_B \}$. We have considered the case where only one coefficient has $m$-degeneracy, but the same arguments can just be repeated as necessary for every coefficient that has degeneracy, which is sufficient to prove the general case.
\end{proof}

Theorem~\ref{thm::purestates} reveals that for every coherence measure and pure bipartite state, then there is always a basis where the coherence exactly quantifies the entanglement. In a more practical sense, it also shows that for every coherence measure that is computable, there corresponds a computable entanglement measure for pure states. Previously, we have already seen that the relative entropy of coherence, which has a closed form expression, corresponds to the entropy of entanglement for pure states. We can similarly choose the $l1$ norm of coherence, for which we get the simple closed form formula $E_{\mathcal{C}}(\ket{\psi}_{AB}) = \sum_{i \neq j } \sqrt{\lambda_i\lambda_j}$ where $\ket{\psi}_{AB} = \sum_i\sqrt{\lambda_i}\ket{i,i}_{AB}$. Theorem~\ref{thm::purestates} states that this expression is also a valid entanglement monotone for pure states. In general, there exists an infinite number of computable coherence measures. We also note that once one has an entanglement monotone for pure states, then it is possible to generalize it to mixed states via a convex roof construction~\cite{Horodecki2001}, which provides yet another avenue for generating new entanglement measures from coherence measures.

\section{Asymmetric quantifiers of quantum correlations}

In the previous section, the symmetric portion of the correlated coherence was considered, in which case it was found to directly address the entangled part of quantum correlations. We now show that simply dropping the requirement of symmetry naturally leads to discord-like measures of correlations.

For quantum discord, the set of states that has zero discord, and are thus ``classical", are the set of classical-quantum states which can be written in the form $\rho_{AB}=\sum_i p_i \ket{i}_A\bra{i}\otimes \rho_B^i$. One may readily define this set of classical quantum states by considering extensions without the requirement that the extension is symmetric. Let us consider the following:

\begin{definition} [Asymmetric discord of coherence]
The asymmetric discord of coherence, for any given coherence measure $\mathcal{C}$, is defined to be the following quantity: $$D_{\mathcal{C}}(\rho_{AB}) = \min_{B'}\mathcal{C}_{\mathrm{min}}(A:BB' \mid \rho_{ABB'})$$ where the minimization is performed over all possible extensions satisfying $\Tr_{B'}(\rho_{ABB'}) = \rho_{AB}$.
\end{definition}

We can then observe that this always defines a discord-like quantifier for every coherence measure $\mathcal{C}$.

\begin{theorem}
$D_{\mathcal{C}}(\rho_{AB}) = 0$ iff $\rho_{AB}$ is classical-quantum, i.e. the state can be written as $\rho_{AB}= \sum_i p_i \ket{i}_A\bra{i}\otimes \rho_B^i$ where $\{\ket{i}_A\}$ is some orthonormal set. It is strictly positive otherwise.
\end{theorem}

\begin{proof}
First, suppose $\rho_{AB}= \sum_i p_i \ket{i}_A\bra{i}\otimes \rho_B^i$. Writing $\rho_B^i$ in terms of its pure state decomposition, we have $$\rho_{AB}= \sum_i p_i \ket{i}_A\bra{i}\otimes \sum_j q_{ij}\ket{\beta_j}_B\bra{\beta_j}.$$ This state always permits an extension on Bob's side of the form $$\rho_{ABB'}= \sum_i p_i \ket{i}_A\bra{i}\otimes \sum_j q_{ij}\ket{\beta_j}_B\bra{\beta_j}\otimes \ket{i,j}_{B'}\bra{i,j}$$ for which $\mathcal{C}_{\mathrm{min}}(A:BB' \mid \rho_{ABB'})=0$ and so $D_{\mathcal{C}}(\rho_{AB})=0$.

Conversely, if $D_{\mathcal{C}}(\rho_{AB})=0$ then this implies that we can write $\rho_{ABB'} = \sum_i p_i \ket{i}_A\bra{i} \otimes \ket{\beta_i}_{BB'}\bra{\beta_i}$, is is a classical-quantum state and will remain classical-quantum even if we trace out the subsystem $B'$. This proves the converse statement so we must have  $D_{\mathcal{C}}(\rho_{AB}) = 0$ iff $\rho_{AB}$ is classical-quantum. 

Since $D_{\mathcal{C}}(\rho_{AB})$ is a coherence measure and so is nonnegative, and $D_{\mathcal{C}}(\rho_{AB}) = 0$ iff $\rho_{AB}$ is classical-quantum, we must have that for any non classical-quantum state,  $D_{\mathcal{C}}(\rho_{AB})> 0$. This completes the proof.
  
\end{proof}

The most general notion of nonclassical correlations is one where the set of classical states is the set of classical-classical states, or completely classical states. These are quantum states that can always be written in the form $\rho_{AB}=\sum_{i,j} p_{ij} \ket{i}_A\bra{i}\otimes \ket{j}_B\bra{j}$. This can be directly addressed via the correlated coherence itself, without consideration of any extensions of states, which is the natural end point of the further relaxation of the constraints that were previously considered in $E_{\mathcal{C}}$ and $D_\mathcal{C}$

\begin{theorem}
$\mathcal{C}_{\mathrm{min}}(A:B \mid \rho_{AB}) = 0$ iff $\rho_{AB}$ is classical-classical, i.e. the state can be written as $\rho_{AB}=\sum_{i,j} p_{ij} \ket{i}_A\bra{i}\otimes \ket{j}_B\bra{j}$ where $\{\ket{i}_A\}$ and $\{\ket{j}_B\}$ are some orthonormal sets. It is strictly positive otherwise.
\end{theorem}

\begin{proof}
First, suppose $\rho_{AB}=\sum_{i,j} p_{ij} \ket{i}_A\bra{i}\otimes \ket{j}_B\bra{j}$. It is then immediate clear by considering the basis $\{ \ket{i}_A \ket{j}_B$ that $\mathcal{C}_{\mathrm{min}}(A:B \mid \rho_{AB}) = 0$.

Conversely, if $\mathcal{C}_{\mathrm{min}}(A:B \mid \rho_{AB}) = 0$ then this implies that we can write $\rho_{AB} = \sum_{i,j} p_{ij} \ket{i}_A\bra{i} \otimes \ket{j}_B\bra{j}$ since there must be some local basis $\{ \ket{i}_A \}$ and $\{ \ket{j}_B \}$ for which $\rho_{AB}$ is diagonal. This proves the converse statement so we must have  $\mathcal{C}_{\mathrm{min}}(A:B \mid \rho_{AB}) = 0$ iff $\rho_{AB}$ is classical-classical. 

Since $\mathcal{C}_{\mathrm{min}}(A:B \mid \rho_{AB})$ is a coherence measure and so is nonnegative, and $\mathcal{C}_{\mathrm{min}}(A:B \mid \rho_{AB}) = 0$ iff $\rho_{AB}$ is classical-classical, so we must have that for any non classical-classical state,  $\mathcal{C}_{\mathrm{min}}(A:B \mid \rho_{AB})> 0$. This completes the proof.
  
\end{proof}

We also observe that for pure bipartite states, $\mathcal{C}_{\mathrm{min}}(A:B \mid \ket{\psi}_{AB}) = D_{\mathcal{C}}(\ket{\psi}_{AB})=E_{\mathcal{C}}(\ket{\psi}_{AB})$, The discord-like quantifiers converge with entanglement, which is a known property of measures of quantum discord.
\section{Conclusion}

In the preceding sections, we presented a construction that is valid quantifier of entanglement. The construction is also generalizable to include larger classes of quantum correlations, leading to discord-like quantifiers of nonclassicality. The arguments are independent of not only the type of coherence measure used, it is also independent of the kind of non-coherence producing operation that is being considered. Such entanglement measures must therefore necessarily exist for any convex coherence quantifier that shares a common notion of classicality.  This leads to the conclusion that such constructions, and thus notions of entanglement and discord, must exist in every reasonable resource theory of coherence.

In~\cite{Streltsov2015}, it was demonstrated that for every entanglement measure, there corresponds a coherence measure. This was achieved by considering the entanglement of the state after performing some preprocessing in the form an an incoherent operation. In a sense, this article asks the converse question: Does every coherence measure correspond to some entanglement measure? The results discussed in this article proves this in the affirmative. Therefore, if one were interested in keeping count, the number of possible entanglement measures must be exactly equal to the number of coherence measures. 

The fact that entanglement can always be defined as the symmetric portion of correlated coherence also further illuminates the role that is being played by the incoherent operation in~\cite{Streltsov2015}, despite not being a crucial element for the construction of entanglement measures. Recall that incoherent operations are operations that do not produce coherence. This does not, however, preclude the moving of coherence from one portion of the Hilbert space to another. Since coherence can always be faithfully  convert coherence into entanglement, we see that the incoherent operation, in such a context, is performing the role of converting any local coherences into the symmetric portion of the correlated coherence, at least when one restrict themselves to the resource theory of coherence considered in~\cite{Baumgratz2014,Streltsov2015 }.

We hope that the discussion presented here will inspire further research into the interplay between coherence and quantum correlations.

\acknowledgements This work was supported by the National Research Foundation of Korea (NRF) through a grant funded by the Korea government (MSIP) (Grant No. 2010-0018295) and by the Korea Institute of Science and Technology Institutional Program (Project No. 2E27800-18-P043). K.C. Tan was supported by Korea Research Fellowship Program through the National Research Foundation of Korea (NRF) funded by the Ministry of Science and ICT (Grant No. 2016H1D3A1938100). We would also like to thank H. Kwon for helpful discussions.


\begin{thebibliography}{99}

\bibitem{Einstein1935} A. Einstein, B. Podolsky, and N. Rosen, Phys. Rev. {\bf 47}, 777 (1935).

\bibitem{Werner1989} R. F. Werner, Phys. Rev. A {\bf 40}, 4277 (1989).

\bibitem{Horodecki2001} M. Horodecki, Quantum Inf. Comput. 1, 1 (2001).

\bibitem{Ekert1991} A. K. Ekert, Phys. Rev. Lett. {\bf 67}, 661 (1991).

\bibitem{Bennett1991} C. H. Bennett, G. Brassard, C. Crepeau, R. Jozsa, A. Peres, and W. K. Wootters, Phys. Rev. Lett. {\bf 67}, 661 (1991).

\bibitem{Bennett1992} C. H. Bennett and S. J. Wiesner, Phys. Rev. Lett. {\bf 69}, 2881 (1992).


\bibitem{Ollivier2001} H. Ollivier and W. H. Zurek, Phys. Rev. Lett. {\bf 88}, 017901 (2001).

\bibitem{Henderson2001} L. Henderson and V. Vedral, J. Phys. A {\bf 34}, 6899 (2001).

\bibitem{Datta2008} A. Datta, A. Shaji,and C. M. Caves, Phys. Rev. Lett. {\bf 100}, 050502 (2008).

\bibitem{Chuan2012} T. K. Chuan, J. Maillard, K. Modi, T. Paterek, M. Paternostro, M. Piani, Phys. Rev. Lett. {\bf 109}, 070501 (2012).

\bibitem{Dakic2012} B. Dakic, Y. O. Lipp, X. Ma, M. Ringbauer, S. Kropatschek,
S. Barz, T. Paterek, V. Vedral, A. Zeilinger, C. Brukner, and P. Walther, Nat. Phys. {\bf 8}, 666 (2012).

\bibitem{Chuan2013} T. K. Chuan and T. Paterek, New J. Phys. {\bf 16}, 093063 (2013).

\bibitem{Aberg2006} J. Aberg, arXiv:quant-ph/0612146 (2006).

\bibitem{Baumgratz2014} T. Baumgratz, M. Cramer, and M. Plenio, Phys. Rev. Lett. {\bf 113}, 140401 (2014).

\bibitem{Levi2014} F. Levi and F. Mintert, New J. Phys. {\bf 16}, 033007 (2014).

\bibitem{Streltsov2015} A. Streltsov, U. Singh, H. S. Dhar, M. N. Bera, and G. Adesso, Phys. Rev. Lett. 115, 020403 (2015).

\bibitem{Tan2016} K.C. Tan, H. Kwon, C.-Y. Park and H. Jeong, Phys. Rev. A {\bf 94}, 022329 (2016).

\bibitem{Ma2016} J. Ma, B. Yadin, D. Girolami, V. Vedral, and M. Gu, Phys. Rev. Lett. {\bf 116}, 160407 (2016).

\bibitem{Tan2018} K. C. Tan, S. Choi, H. Kwon and H. Jeong, arXiv:1711.07185 (2017).

\bibitem{Yadin2016}B. Yadin and V. Vedral, Phys. Rev. A {\bf 93}, 022122 (2016).

\bibitem{Kwon2017} H. Kwon, C.-Y. Park, K. C. Tan, H. Jeong, New J. Phys. {\bf 19}, 043024 (2017).

\bibitem{Hillery2016} M. Hillery Phys. Rev. A {\bf 93}, 012111 (2016).

\bibitem{Matera2016} J. M. Matera, D. Egloff, N. Killoran and M.B. Plenio Quant. Sci. Tech. {\bf 1}(1), 01LT01 (2016).

\bibitem{YTWang2017} Y.-T. Wang, J.-S. Tang, Z.-Y. Wei, S. Yu, Z-J. Ke, X.-Y. Xu, C-F. Li and G.-C. Guo, Phys. Rev. Lett. {\bf 118}, 020403 (2017).

\bibitem{Zhang2015} Y-R. Zhang, L-H. Shao, Y. Li and H. Fan, Phys. Rev. A {\bf 93}, 012334 (2016).

\bibitem{Xu2016} J. Xu, Phys. Rev. A {\bf 93}, 032111 (2016).

\bibitem{Tan2017} K. C. Tan, T. Volkoff, H. Kwon and H. Jeong, Phys. Rev. Lett. {\bf 119}, 190405 (2017).

\bibitem{Streltsov2017} A. Streltsov, G. Adesso, and M. B. Plenio, Rev. Mod. Phys {\bf 89}, 041003 (2017).

\bibitem{Chitambar2016} E. Chitambar, M.-H. Hsieh, Phys. Rev. Lett. {\bf 117}, 020402 (2016).

\bibitem{Streltsov2016} A. Streltsov, E. Chitambar, S. Rana, M.N. Bera, A. Winter, and M. Lewenstein, Phys. Rev. Lett. {\bf 116}, 020405 (2016).

\bibitem{Wang2017} X.-L. Wang, Q.-L. Yue, C.-H. Yu, F. Gao and S.-J.Qin, Sci. Rep. 7, 12122 (2017).

\bibitem{Kraft2018} T. Kraft and M. Piani, arXiv:1801.03919 (2018).

\bibitem{Ma2018} T. Ma, M.-J. Zhao, S.-M. Fei and M.-H. Yung, arXiv:1802.08821 (2018).

\bibitem{Nielsen2010}M.A. Nielsen and I.L. Chuang, {\it Quantum Computation and Quantum Information} (Cambridge University Press, Cambridge, England, 2010).

\end{thebibliography}
\end{document}